\documentclass[a4paper,11pt,final]{article}

\usepackage{graphicx}
\usepackage{caption,color}
\usepackage{amsmath,amssymb,amsthm,enumerate}
\usepackage{fullpage}
\usepackage{nicefrac}
\usepackage{xspace}
\usepackage{commath}
\usepackage{verbatim}
\usepackage{hyperref}
\hypersetup{colorlinks=true}
\hypersetup{linkcolor=black,anchorcolor=black,citecolor=black,urlcolor=black}
\hypersetup{pdftitle=Halfspace Chasing}
\hypersetup{pdfauthor={Nikhil Bansal} {Martin Bohm} {Marek Elias} {Grigorios
Koumoutsos} {Seeun William Umboh}}
\hypersetup{pdfkeywords={halfspace chasing problem} {online algorithms}
	{competitive analysis}}

\newtheorem{theorem}{Theorem}
\newtheorem{definition}[theorem]{Definition}
\newtheorem{prop}[theorem]{Proposition}

\newtheorem{lemma}[theorem]{Lemma}

\newtheorem{claim}[theorem]{Claim}

\usepackage{caption}
\usepackage{algorithmicx}
\usepackage[noend]{algpseudocode}
\usepackage{algorithm}

\usepackage{thmtools,thm-restate}                                   
\newcommand{\ignore}[1]{} 
\usepackage[compact]{titlesec}

\DeclareMathOperator{\OPT}{OPT}

\DeclareMathOperator{\chase}{Chase}

\newcommand{\R}{\mathbb{R}}
\newcommand{\ov}{\mathbf{0}}               
\algnewcommand{\IIf}[1]{\State\algorithmicif\ #1\ \algorithmicthen}
\algnewcommand{\EndIIf}{\unskip\ \algorithmicend\ \algorithmicif}

\newcommand{\twonorm}[1]{\left\lVert #1 \right\rVert_2}
\newcommand{\onenorm}[1]{\left\lVert #1 \right\rVert_1}

\newcounter{note}

\begin{document}

\title{Nested Convex Bodies are Chaseable\thanks{This work was supported by NWO grant 639.022.211, ERC consolidator grant 617951 and GA \v{C}R grant 17-09142S.}}

\author{
Nikhil Bansal\thanks{TU Eindhoven, Netherlands.
\texttt{\{n.bansal,m.elias,g.koumoutsos\}@tue.nl, seeun.umboh@gmail.com}}
\and
Martin B\"ohm\thanks{Charles University, Czech Republic.
	\texttt{bohm@iuuk.mff.cuni.cz}}
\and
Marek Eli\'a\v{s}\footnotemark[2]
\and
Grigorios Koumoutsos\footnotemark[2]
\and
Seeun William Umboh\footnotemark[2]}
\begin{titlepage}
\def\thepage{}
\thispagestyle{empty}

\maketitle
\begin{abstract}
  In the Convex Body Chasing problem, we are given an initial point $v_0 \in \R^d$ and an online sequence of $n$ convex bodies $F_1, \ldots, F_n$. When we receive $F_i$, we are required to move inside $F_i$. Our goal is to minimize the total distance travelled. 
This fundamental online problem was first studied by Friedman and Linial (DCG 1993). They proved an $\Omega(\sqrt{d})$ lower bound on the competitive ratio, and conjectured that a competitive ratio depending only on $d$ is possible. However, despite much interest in the problem, the conjecture remains wide open.

We consider the setting in which the convex bodies are nested: $F_1 \supset \ldots \supset F_n$. The nested setting is closely related to extending the online LP framework of Buchbinder and Naor (ESA 2005) to arbitrary linear constraints. Moreover, this setting retains much of the difficulty of the general setting and captures an essential obstacle in resolving Friedman and Linial's conjecture. In this work, we give the first $f(d)$-competitive algorithm for chasing nested convex bodies in $\R^d$. 
\end{abstract}
\end{titlepage}

\newpage
\section{Introduction}\label{sec:intro}
In the \emph{convex body chasing} problem, introduced by Friedman and
Linial~\cite{Friedman93}, we are given an initial position $v_0 \in \R^d$. 
At each time step $i$, we receive a convex set $F_i \subset \R^d$ as a request, and to serve the request, we must move to some point $v_i \in F_i$. The goal is to minimize the
total distance traveled to serve the requests.
The distance can be measured using an arbitrary norm, but unless stated
otherwise, it is measured using the Euclidean norm. 
As any convex body can be approximated arbitrarily well by intersection of halfspaces, one can assume that $F_i$ are halfspaces\footnote{If $F$ is the intersection of halfspaces $H_1, \ldots, H_s$, to simulate the request for $F$, the adversary can give $H_1,\ldots,H_s$ several times in a round-robin manner until the online algorithm moves inside $F$. Not revealing $F$ directly can only hurt the online algorithm and does not affect the offline solution.} and hence this problem is also known as \emph{halfspace chasing}.

This problem belongs to a very rich class of problems called
{\em Metrical Service Systems} (MSS) \cite{BLS92}. In an MSS, we are given an arbitrary metric space $(V, \rho)$ and an initial position $v_0\in V$. At each time $i$,
a request set $F_i \subset V$ arrives and we must serve it by
moving to some $v_i \in F_i$.
MSS captures several interesting online problems such as the $k$-server problem. While almost tight bounds are known for general MSS~\cite{BLS92,FM00,BBM01}, these bounds are not so interesting as typical online problems correspond to
MSS with highly structured requests $F_i$ and metric space $(V,\rho)$.
There has been a lot of interesting work on particular cases of MSS, e.g.~\cite{KP95,CL96,B96,BK04,SS06,Sit14}, 
but understanding the role of structure in MSS instances is a major long-term goal in online computation with far-reaching consequences.

Indeed, the main motivation of~\cite{Friedman93} for considering the convex body chasing problem was to express the competitive
ratio of MSS in terms of geometric properties of the request sets $F_i$. For the convex body chasing problem, they obtained an $O(1)$-competitive algorithm for $d=2$; for $d >2$, they gave an $\Omega(\sqrt{d})$ lower bound and conjectured that a  competitive ratio depending only on $d$ is possible. However, despite much interest, the conjecture remains open.

\paragraph{Nested Convex Body Chasing.}
In this paper, we consider the {\em nested convex body chasing} problem 
where the requested convex bodies are nested, i.e.,~$F_i \subset F_{i-1}$ for each $i\geq 1$. 
This natural special case is closely related to many fundamental questions in online algorithms and online learning,
and has been of interest in recent years.
However, prior to our work, nothing was known for it beyond the results of Friedman and Linial~\cite{Friedman93}
for the general case.

\subsection{Connections and Related Work}
A useful equivalent formulation of the nested problem is the following:
Given an initial position $v_0$, at each time step $i$, we receive some arbitrary convex body $F_i$ (not necessarily nested), and we must move to some point $v_i$ that is contained in \emph{every} convex body seen so far, i.e.~$v_i \in F_1 \cap \ldots \cap F_i$. The goal is to minimize the total distance traveled.
Indeed, this is equivalent to convex body chasing with requests
$F'_i= F_1 \cap \ldots \cap F_i$, which form a nested sequence.

\paragraph{Online Covering LP.} The influential primal-dual framework of
Buchbinder and Naor for online covering LPs \cite{BN09} can now be seen as a
special case of nested convex body chasing with the $\ell_1$-norm. In the former problem, the algorithm
starts at the origin $v_0=\bf{0}$, and at each time $i$,
a linear constraint $a_i^\top x \geq b_i$ with non-negative $a_i$ and $b_i$ arrives. 
The goal is to maintain a feasible point $x_i$
that satisfies all previous requests while the coordinates of $x_i$ are only
allowed to increase over time.
The objective function $c^\top x$ (where $c$ is non-negative) can be assumed to be $\onenorm{x}$ by scaling.
Finally, note that in nested convex body chasing with covering constraints, it never helps to decrease any variable and hence online covering LP is indeed a special case of nested convex body chasing.

While the online primal-dual framework \cite{BN09} has been applied successfully to many online problems, so far it is limited to LPs with covering and packing constraints, and minor tweaks thereof. An important goal is to extend the online LP framework more broadly beyond packing and covering LPs. For example, it is unclear how to do this even for seemingly simple formulations such as $k$-server on depth-$2$ HSTs or Metrical Task Systems on a line. 
Since the nested convex body chasing problem corresponds to solving online LPs with arbitrary constraints (with both positive and negative entries) and a specific type of objective, understanding the nested convex body chasing problem is an essential step towards this goal. Indeed, this is one of our main motivations to consider this problem. 

\paragraph{General Convex Body Chasing.} Another motivation for studying the
nested case is that it captures much of the inherent hardness of the general
convex body chasing problem. For example, the $\Omega(\sqrt{d})$ lower
bound~\cite{Friedman93} for the general problem also holds in the nested
setting. Moreover, several natural algorithms also fail for the nested case.

\paragraph{Other Special Cases.} The only known algorithms for chasing convex
bodies in $\R^d$ with $d > 2$ are for certain restricted families of convex
bodies $F_i$ such as lines and affine subspaces. For chasing lines, Friedman and
Linial~\cite{Friedman93} gave an $O(1)$-competitive algorithm. For chasing lines and half-line, Sitters~\cite{Sit14} showed that the generalized work function algorithm (WFA) is also $O(1)$-competitive; this is interesting as the WFA is a generic algorithm that attains nearly-optimal competitive ratios for many MSS and is a natural candidate to be $f(d)$-competitive for convex body chasing. 
Recently, Antoniadis et al.~\cite{Antoniadis16} gave an elegant and simple $O(1)$-competitive algorithm 
for chasing lines, and a $2^{O(d)}$-competitive algorithm for chasing affine subspaces.
However, all these results crucially rely on the fact that
the requests $F_i$ have a lower dimension and 
do not seem to apply to
our problem.

\paragraph{Connections to Online Learning.} The convex body chasing problem is also closely related to recent work combining aspects of competitive analysis and online learning. One such work is the \emph{Smoothed Online Convex Optimization} setting of Andrew et al.~\cite{ABLLMRW13, ABLLMRW15} which incorporates movement cost into the well-studied online learning setting of online convex optimization. The problem is well-understood for $d=1$ \cite{BGKPSS15,AS17}, but nothing is known for larger $d$. Another related work is that of Buchbinder, Chen and Naor~\cite{BCN14} which combines online covering LPs with movement cost.

\subsection{Our Results}
Our main result is the following.
\begin{theorem}
  \label{thm:alg}
  There is an algorithm for chasing nested convex bodies in $\R^d$ with
competitive ratio that only depends on $d$. In particular, it has competitive
ratio $O(6^d(d!)^2)$. 
\end{theorem}
The algorithm is described in Section \ref{sec:alg} and is based on two ideas.
First we show that to design an $O_d(1)$-competitive algorithm for chasing
nested convex bodies, it suffices to make an algorithm for $r$-bounded
instances,
where all the bodies $F_i$ are completely contained in some ball $B(v,r)$ with
radius $r$ and center $v$. 
Moreover, even though competitive ratio is a relative guarantee, it suffices to
bound the total movement of the algorithm on any $r$-bounded instance by
$O_d(r)$. 
Proving an absolute bound on the distance moved makes the algorithmic task
easier and we design such a bounded chasing algorithm
in Section \ref{sec:alg_bounded}. 

Surprisingly, the natural approaches for $r$-bounded instances based on the Ellipsoid Method or
the centroid approach do not work. In particular, consider a $1$-bounded
instance where the initial body is $F_1= B(\ov,1)$,
and the algorithm starts at the origin. As nested convex bodies arrive, if the
current point $v_{i-1}$ is infeasible for the request $F_i$, a natural approach might be to
move to the centroid of $F_i$ or to the center of the minimum volume ellipsoid
enclosing $F_i$ (see Figure \ref{fig:intro_ellipsoid}). 
In Section~\ref{sec:lb}, we describe a simple $1$-bounded instance in $\R^2$,
where the above algorithms travel an unbounded distance. 

\begin{figure}
\hfill\includegraphics[scale=0.9]{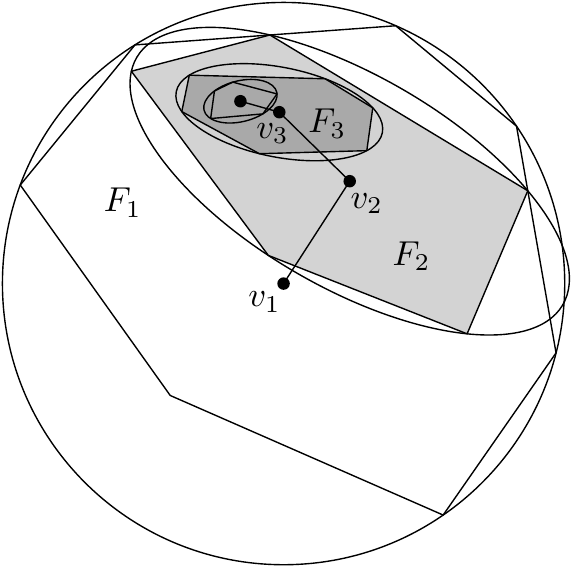}\hfill\ \caption{A few steps of the ellipsoid algorithm.}
\label{fig:intro_ellipsoid}
\end{figure}

We design our $d$-dimensional bounded chasing
algorithm in Section~\ref{sec:alg} based on a recursive approach together with some
simple geometric properties. It iteratively invokes the
$(d-1)$-dimensional algorithm on at most $d$ bounded instances defined on some
suitably chosen hyperplanes. When these instances end, we can argue that the
future requests must lie in some smaller ball $B(v',\gamma r)$, for some fixed
$\gamma < 1$. 
Roughly, this allows us to bound the competitive ratio by $g(d)$, that satisfies
the recurrence $g(d)  \leq d g(d-1)/(1-\gamma)$.

\section{Preliminaries}\label{sec:prelim}
We define some notation and recall some basic facts from geometry.

\begin{definition}[Nested Convex Body Chasing]
  In the \emph{nested convex body chasing problem} in $\R^d$, the algorithm starts
  at some position $v_0$, and an online sequence of $n$ nested convex
  bodies $F_1 \supset \ldots \supset F_n$ arrive one by one. When
  convex body $F_i$ arrives, the algorithm must move to a point $v_i$
  that lies in $F_i$. The goal is to minimize the total distance
  traveled $\sum_{i=1}^n \twonorm{v_i - v_{i-1}}$.
\end{definition}
Note that the choice of measuring distance using the $\ell_2$-norm---as opposed
to some other symmetric norm, say the $\ell_1$-norm---has a negligible effect on
the competitive ratio that we obtain because all symmetric norms on $\R^d$ are
within a $d^{1/2}$ factor of each other.

Let $B(v,r)$ denote the ball of radius $r$ centered at $v$.
The following useful fact is a variant of John's theorem for balls instead
of ellipsoids (for a proof, see e.g. \cite[Lemma 8.7.3]{MG07}):

\begin{prop}[Minimum-volume enclosing ball]
  \label{prop:ball}
  Let $F$ be a bounded convex body and suppose $B(v,r)$ is the minimum-volume ball enclosing $F$. Then, the center $v$ of the ball $B(v,r)$ is contained in $F$.
\end{prop}

Next we need the following standard fact that we prove here for completeness. We will use it to show that either we can reduce to a $(d-1)$-dimensional instance or a $d$-dimensional instance that is contained in a ball with smaller radius.
We use $\ov$ to denote the origin.

\begin{prop}[Dimension reduction or radius reduction]
  \label{prop:reduction}
Let $d \geq 2$ and $F$ be a bounded convex body in $\R^d$ contained in
$B(\ov,r)$. Then, either $F$ intersects some axis-aligned hyperplane,
or it is contained in some orthant of $B(\ov,r)$. Moreover, in the latter case, the smallest ball $B(s,r')$ enclosing $F$ has radius $r' \leq r (1 - 1/d)^{1/2}$.
\end{prop}

\begin{proof}
The first part immediately follows by convexity. For the second part, we
assume, without loss of generality, that $r = 1$ and that $F$ is contained in
the positive orthant of $B(\ov,1)$. We now show that every point $x$ in the
positive orthant of $B(\ov,1)$ is within distance at most $r':=(1-1/d)^{1/2}$ from the point $s = (1/d, \ldots, 1/d)$. There are two cases: (1) $\onenorm{x} \leq 1$; (2) $\onenorm{x} > 1$.

In case (1), $x$ is in the convex hull of $e_1, \dotsc, e_d$, and $\ov$,
where $e_k$ denotes the $k$-th vector of the standard basis with $1$ in the $k$-th
coordinate and $0$ elsewhere.
Therefore, it suffices to prove that $\ov$ and $e_1, \dotsc, e_d$ are within distance $r'$ from $s$.
Indeed, $\twonorm{s - \ov} = (1/d)^{1/2} \leq r'$ (as $d \geq 2$) and
$\twonorm{s - e_k}=  r'$ for each $k \in [d]$.

In case (2), we have 
     \[ \twonorm{s - x}^2 = \sum_{k=1}^d \left(x_k - \frac{1}{d}\right)^2 = \sum_{k=1}^d \left(x_k^2 - \frac{2x_k}{d} + \frac{1}{d^2}\right) = \twonorm{x}^2 - \frac{2\onenorm{x}}{d} + \frac{1}{d} \leq 1 - 1/d, \] 
                                    where the inequality uses that $x \in B(\ov,1)$ and hence $\twonorm{x}\leq 1$ and $\onenorm{x} > 1$. 
		Thus, the positive orthant of $B(\ov,1)$ is contained in $B(s,r (1-1/d)^{1/2})$.
\end{proof}

\section{Algorithm}\label{sec:alg}

We now describe our algorithm and prove Theorem \ref{thm:alg}.

We first show, using a guess-and-double approach, that finding a good algorithm for the general nested convex body chasing problem can be reduced to an
easier problem of designing an algorithm for which we can upper bound the {\em absolute} distance traveled, on bounded instances of the following type.

\begin{definition}[$r$-Bounded Instances]
An instance with starting point $v_0$ and requests $F_1 \supset \ldots \supset F_n$ is said to be \emph{$r$-bounded} if every request $F_i$ is contained in $B(v_0,r)$.
\end{definition}

Note that a general instance may not be $r$-bounded for any finite $r$, e.g.~in a covering LP where all the $F_i$ are halfspaces of the type $a_i^\top x \geq 1$ where $a_i$ has all entries non-negative. 

As we shall see, the task of showing an absolute bound on the distance traveled (instead of a relative bound needed for competitive ratio) makes the problem cleaner. We now describe the reduction.

\begin{lemma}[Reduction to Bounded Chasing]
  \label{lem:reduction}
 For any fixed $r > 0$, if there exists an algorithm $\chase_d$ for $r$-bounded instances that travels a total distance of at most $g(d) \cdot r$, then there exists a $f(d)$-competitive algorithm for general instances with $f(d) = 4(g(d) + 1)$.
\end{lemma}
\begin{proof}
Consider a general instance with  starting point $v_0$. Let $\delta_i$ be the distance between $v_0$ and the closest point in $F_i$; note that $\delta_i$ is non-decreasing in $i$ because $F_i$'s are nested. Wlog, we can assume that $v_0 \notin F_1$ and $\delta_1 = 1$ (by scaling).

The algorithm for the general instance proceeds in \emph{stages}. For $j=1,2,\ldots$, stage $j$ consists of all requests $F_i$ for which $\delta_i \in [2^{j-1}, 2^j)$, i.e.,~stage $j$ begins with the first request $F_i$ that intersects with $B(v_0,2^{j})$ but not with $B(v_0,2^{j-1})$, and ends with the last request $F_{i'}$ that intersects with $B(v_0,2^j)$.

The algorithm will run a new instance of $\chase_d$  at each stage $j$. Let $F_{s(j)}$ be the first request of stage $j$ and $F_{\ell(j)}$ be the last. At the start of stage $j$, the algorithm starts at the point $v_0$, and begins an instantiation $\chase_d(v_0,2^j)$ of $\chase_d$
that it runs over the course of the stage with requests
\[F_{s(j)} \cap B(v_0,2^j), F_{s(j)+1} \cap B(v_0, 2^j), \ldots, F_{\ell(j)} \cap B(v_0,2^j).\]
Note that these requests form a $2^j$-bounded instance.

We now bound the performance of the algorithm. 
Clearly, $\OPT =  \delta_n$. Let $j^*$ denote the index of the final stage, and hence $\OPT \geq 2^{j^*-1}$.
   For each $1 \leq j \leq j^*$, the movement cost during stage $j$ has two parts and can be bounded as follows: 
\begin{itemize}
\item The movement of $\chase_d(v_0,2^j)$. This is at most $g(d) \cdot 2^j$ by the assumed guarantee on $\chase_d$. 
\item Returning to $v_0$ from its previous location, just before stage $j$ begins. This cost is at most $2^{j'}\leq 2^{j-1}$ where $j'<j$ is the stage that ends just before stage $j$. This is because the algorithm was following $\chase_d(v_0,2^{j'})$ which always stays within $B(v_0,2^{j'})$ (as all $\chase_d(v_0,r)$ requests lie in $B(v_0,r)$).
\end{itemize}
So the total distance traveled by our algorithm is at most 
 \[\sum_{j \leq j^*} 2^j(g(d) + 1) \leq 2^{j^*+1} (g(d) + 1) \leq 4(g(d)+1)\OPT. \qedhere \]
\end{proof}

\subsection{Bounded Chasing Algorithm}
\label{sec:alg_bounded}
We now focus on designing an algorithm for the Bounded Chasing problem. The following theorem is our main technical result.
 
\begin{theorem}[Bounded Chasing Theorem]
  \label{lem:bounded-chasing}
  There exists an algorithm $\chase_d$ that travels at most $g(d)\cdot r$ distance on $r$-bounded instances where $g(d) = 6^d(d!)^2$.
\end{theorem}

Before we prove Theorem \ref{lem:bounded-chasing}, let us note that  Theorem~\ref{thm:alg} immediately follows by 
combining Theorem~\ref{lem:bounded-chasing} and Lemma~\ref{lem:reduction} and noting that $g(d)\geq 1$. 

We now construct the algorithm $\chase_{d}$ and prove Theorem~\ref{lem:bounded-chasing}. The proof is by induction on $d$. The base case ($d = 1$) is trivial: the requests form nested intervals and the greedy algorithm that always moves to the closest feasible point is $1$-competitive, so $g(d) = 1$. In the remainder of this section, we focus on the $d \geq 2$ case and assume that there exists a $(d-1)$-dimensional algorithm $\chase_{d-1}$ with the required properties.

\vspace{3mm}

\paragraph{Algorithm.}
Consider an $r_0$-bounded instance with starting point $s_0$. 
The high level idea of the algorithm is to reduce the instance into a sequence of $(d-1)$-dimensional instances and run $\chase_{d-1}$ on these instances. 

The algorithm  runs in phases. Each phase starts at some center $s$ with radius parameter $r \leq r_0$.
The first phase starts at $s=s_0$ with radius $r=r_0$. 
In each phase, we run $\chase_{d-1}$ with center $s$ and radius $r$ on the $(d-1)$-dimensional instances induced by the $d$ axis-aligned hyperplanes $H_1,\ldots,H_d$ containing $s$. These are called {\em hyperplane steps}.
When some request $F_i$ arrives that does not intersect with any of these hyperplanes $H_1, \ldots, H_d$, we perform a {\em recentering step} by computing the smallest ball $B(s',r')$ enclosing $F_i$ and moving to $s'$; the current phase then ends, and a 
new phase starts with center $s'$ and radius $r'$. A key property we will use in the analysis (based on Proposition~\ref{prop:reduction}) is that $r' \leq (1-1/d)^{1/2} r$, which will allow us to argue that algorithm makes progress.

\paragraph{Description of a phase.}
We now describe how a phase works. The reader may find it helpful to refer to Figure~\ref{fig:alg_2d} while reading the description below.

Consider a phase that starts at center $s$ and radius $r$. For notational convenience, we reindex the requests so that the first request of the phase is $F_1$.  Let $H_1,\ldots,H_d$ denote the axis-aligned hyperplanes passing through $s$.

\vspace{1mm}

{\noindent \bf  Hyperplane Steps.} Initially at request $F_1$, we choose the axis-aligned hyperplane $H_k$ with the smallest index $k \in [d]$ that intersects $F_1$  (if no such hyperplane exists, we move to the Recentering step below), and run $\chase_{d-1}$ on the $(d-1)$-dimensional instance induced by $H_k$ and follow  it  for as long as we can.
 More specifically, we run $\chase_{d-1}$ on the $(d-1)$-dimensional instance with starting point $s$ and radius $r$, and requests
\[F_1 \cap H_k, \ldots, F_{\ell(k)} \cap H_k,\]
where $F_{\ell(k)}$ is the last request in the current phase that intersects $H_k$; for $i \leq \ell(k)$, we serve request $F_i$ by moving to $v_k(i)$ where $v_k(i)$ is the location of $\chase_{d-1}$ on request $F_i \cap H_k$. 

When the first request $F_i$ arrives that does not intersect the current $H_k$, i.e.,~$F_i = F_{\ell(k)+1}$, then we change the hyperplane  and repeat the above process. That is,
we pick $H_{k'}$ that intersects $F_i$, with the smallest index $k'$ (if it exists),  and run $\chase_{d-1}$ on $H_{k'}$ starting at position $s$ with radius $r$ and  requests $F_{\ell(k)+1} \cap H_{k'},\ldots$ and follow it for as long as we can.

\vspace{1mm}

{\noindent \bf Recentering Step.}
If a request $F_i$ arrives that does not intersect any $H_{k}$ for $k\in [d]$, we compute the smallest ball $B(s',r')$ containing $F_i$, move to $s'$ and serve $F_i$ (note that $s' \in F_i$ by Proposition~\ref{prop:ball}).
The current phase ends, and a new phase with center $s'$ and radius $r'$ starts.

This completes the description of a phase  and we now turn to analyzing the algorithm.

\begin{figure}
\hfill\includegraphics[width=\textwidth]{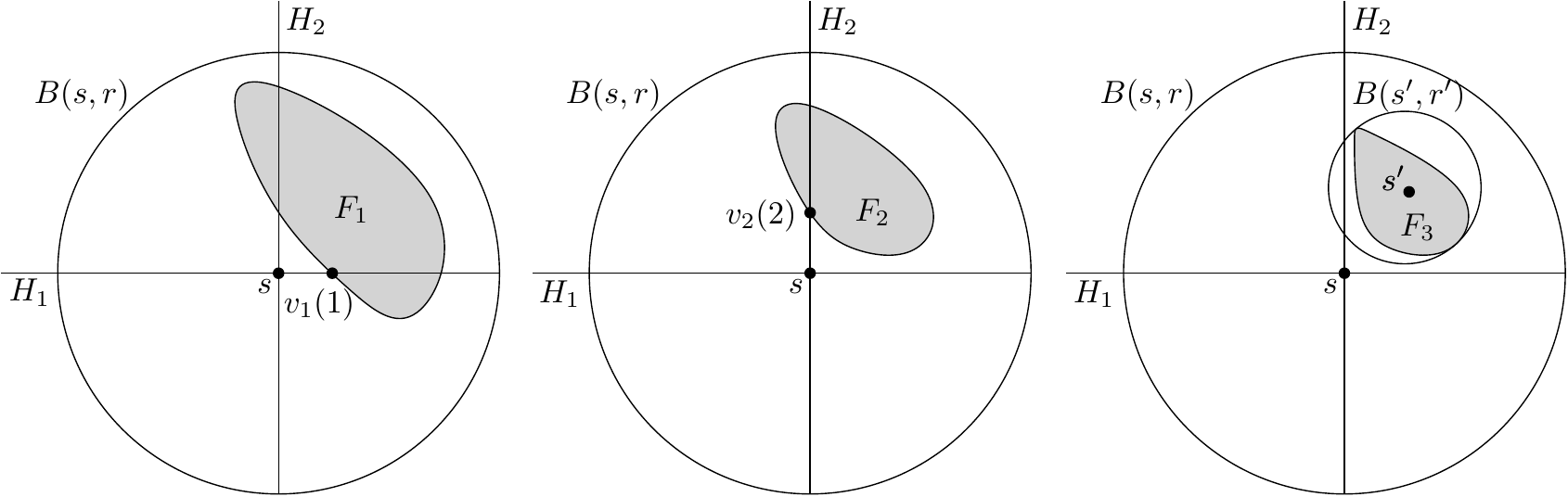}\hfill\ \caption{A phase of $\chase_2$, that starts at $s$. The first request $F_1$ is served using $\chase_1$ in $H_1$.
The second request $F_2$ does not intersect $H_1$ so it is served using $\chase_1$ in $H_2$. Finally, $F_3$ does not intersect $H_2$ either
and a recentering step is performed.}
\label{fig:alg_2d}
\end{figure}

\vspace{3mm}

\paragraph{Analysis.}
We need to show that $\chase_{d}$ is always feasible (Claim \ref{clm:feas}) and bound the distance it travels (Claim \ref{clm:cost}). These claims give us Theorem~\ref{lem:bounded-chasing}.

\begin{claim}[Feasibility]
  \label{clm:feas}
  The algorithm $\chase_{d}$ is always feasible, i.e. $v_i \in F_i$ for all $i$.
\end{claim}

\begin{proof}
We prove the claim by induction on $d$. For $d=1$, the algorithm is trivial and it is always feasible. Assume that the claim is true for $d-1$.
  Consider some request $F_i$. Observe that $\chase_d$ either performs a hyperplane step or a recentering step. In the former, since we follow $\chase_{d-1}$ on some hyperplane $H_k$ and $\chase_{d-1}$ stays feasible (by induction), so we have that $v_i \in F_i \cap H_k$. 
In the latter, $v_i$ is the center of the smallest ball containing $F_i$ so $v_i \in F_i$ by Proposition \ref{prop:ball}. Thus, in both cases, $v_i \in F_i$.
\end{proof}

Next, we bound the distance travelled by $\chase_d$.

\begin{claim}[Cost]
  \label{clm:cost}
  The total distance travelled by $\chase_d$ on an $r$-bounded instance is at most $g(d)\cdot r$, where $g(d) = 6^d(d!)^2$.
\end{claim}
\begin{proof}
  We now bound the distance traveled during each phase. Consider phase $j$. Let $B_j$ denote $\chase_d$'s enclosing ball during the phase and $r_j$ be its radius. Note that during the phase, the algorithm stays within $B_j$ and all requests $F_i$ in the phase are contained in $B_j$. The movement in phase $j$ consists of:
  \begin{itemize}
  \item Movement due to hyperplane steps. On each hyperplane $H_k$, we move at most $g(d-1)\cdot r_j$ by following $\chase_{d-1}$. Thus, the total movement due to hyperplane steps is at most $d\cdot g(d-1)\cdot r_j$
    \item Movement due to switching hyperplanes. We switch hyperplanes at most $d-1$ times, so this is at most $(d-1) \cdot 2r_j$.
  \item Movement due to recentering. This is at most $2r_j$.
  \end{itemize}
  Thus, the total distance traveled in phase $j$ is at most 
  \[d \cdot (g(d-1) + 2) \cdot r_j \leq 3d \cdot g(d-1) \cdot r_j,\] since $g(d-1) \geq 1$.
	
 By Proposition~\ref{prop:ball}, the radii of the enclosing balls decrease geometrically across phases: $r_j \leq r_{j-1} (1 - 1/d)^{1/2}$ for $j > 1$. As $r_1=r$, the total distance traveled by $\chase_d$ over all phases is at most
  \[3d \cdot g(d-1) \cdot \frac{r}{1 - (1 - 1/d)^{1/2}}  \leq 3d \cdot g(d-1) \cdot 2dr = 6d^2g(d-1)r,\]
	where we use that $1-(1-1/d)^{1/2} \geq 1/(2d)$, as $(1-x)^\alpha \leq 1-\alpha x$ for any $x\in[0,1]$ and  $\alpha< 1$.
 
Thus, we conclude that $\chase_d$ travels at most $g(d)\cdot r$, where $g(d) = 6^d(d!)^2$ is the solution to the recurrence $g(d) = 6d^2g(d-1)$ with base case $g(1) = 1$.
\end{proof}

\section{Lower Bounds for Ellipsoid and Centroid}\label{sec:lb}
In this section, we consider some natural ellipsoid-based and centroid-based algorithms for
chasing nested bodies in the $r$-bounded setting, and show that they are 
not competitive. The main reason these algorithms fail is that for (relatively) flat convex bodies, the center of the bounding
ellipsoid, or the centroid, can move by a large distance in directions that do not matter.

Henceforth, for a set $S \subset \R^d$, let $E(S)$ denote the smallest-volume ellipsoid
containing $S$.

\begin{algorithm}[H]
Let $F_t$ be the current bounded convex body on input.

Whenever the current position becomes infeasible:

\hspace{3em}
	Move to the center of $E(F_t)$.
\caption{An ellipsoid-based algorithm}
\label{alg:ellips}
\end{algorithm}

We now construct an $\R^2$ instance in which Algorithm~\ref{alg:ellips} travels an arbitrarily
large distance while the optimal offline cost is constant. In the following, we will use the notation $(x,y)$ for a point in $\R^2$.

The starting point of the instance is $(0,1)$. Each request $F_t$ is
an intersection of four halfspaces $A,B,C,H_t$. The first three halfspaces $A,B,C$
are $y \ge 0, x \ge -1,$ and $x \le 1$,
respectively. The last halfspace $H_t$ will be different for each
$F_t$.

For the first request $F_1$, we set 
\[H_1 = \{(x,y) \;|\, 2y \leq (1-\alpha)x + (1+\alpha)\},\]
for some parameter $\alpha$. 
Note that the boundary of $H_1$ passes through the points $(-1,\alpha)$ and $(1, 1)$, as seen in Figure~\ref{fig:lb_F4}.
The parameter $\alpha$ is chosen so that the center of $E(F_1)$ is strictly to
the right of the $y$ axis, as guaranteed by the following lemma:

\begin{figure}
\hfill\includegraphics{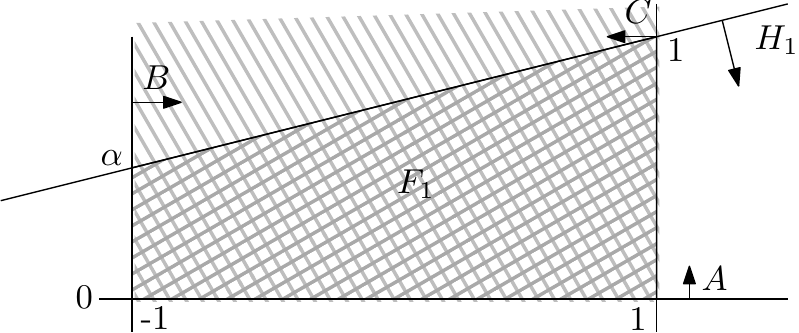}\hfill\ \caption{The first request $F_1$.}
\label{fig:lb_F4}
\end{figure}

\begin{lemma}\label{lem:lb_ellipsoid}
There exists $0<\alpha<1$ such that the center of the smallest ellipsoid
containing $F_1$ has a strictly positive $x$-coordinate. More precisely,
its center is $(c, b)$ with $c,b>0$.
\end{lemma}

We postpone the proof of Lemma~\ref{lem:lb_ellipsoid} and continue with
the description of the request sequence. The remaining nested bodies $\{F_t \mid t \ge 2\}$ are created so that the
$x$-coordinate of the center of $E(F_t)$ oscillates between $c$ and
$-c$. To this end, we construct two infinite families of halfspaces
$R_i$ and $L_i$: for $i \ge 0$, we define
\begin{align*}
R_i &:= \{(x,y)\;|\, 2y\leq (\alpha^{2i}-\alpha^{2i+1})x
	+ (\alpha^{2i}+\alpha^{2i+1})\} \\
L_i &:= \{(x,y)\;|\, 2y\leq (\alpha^{2i+2}-\alpha^{2i+1})x
	+ (\alpha^{2i+2}+\alpha^{2i+1})\}.
\end{align*}
Observe that the boundary of $R_i$ passes through the points $(1, \alpha^{2i})$ and
$(-1,\alpha^{2i+1})$, and the boundary of $L_i$ passes through
$(-1,\alpha^{2i+1})$ and $(1, \alpha^{2i+2})$. See
Figure~\ref{fig:lb_Ri} for an illustration.

\begin{figure}
\hfill\includegraphics{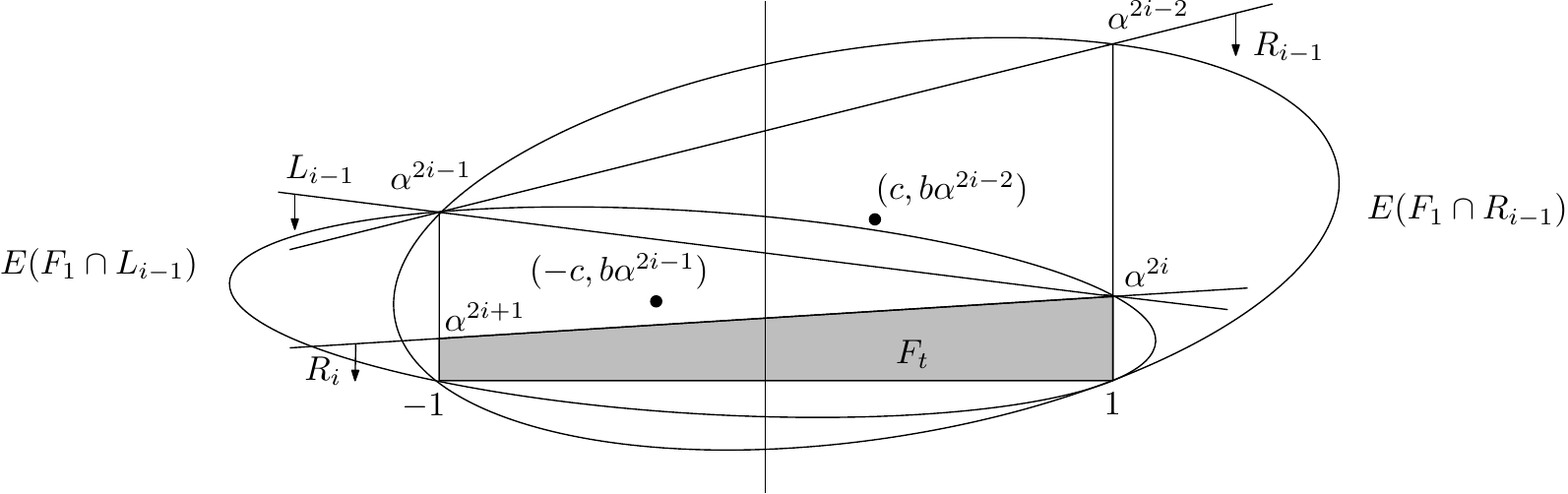}\hfill\ \caption{Halfspaces $L_i$ and $R_i$}
\label{fig:lb_Ri}
\end{figure}

We now describe the requests $F_t$ for $t \geq 2$. For even $t$, we set $H_t$ to be $L_i$, where $i$ is the
smallest index such that $L_i$ does not contain the current position
of the algorithm. For odd $t$, we select $H_t$ to be $R_i$ in a similar fashion. This completes our description of the requests $F_t$. Looking at Figure~\ref{fig:lb_Ri}, one can easily observe that our
requests $F_t$ are indeed nested.

The following lemma describes the position of the center of
each ellipsoid $E(F_t)$. Note that $c$ and $b$ are the strictly positive
constants from Lemma~\ref{lem:lb_ellipsoid}.

\begin{lemma}
When $t$ is odd, the center of $E(F_t)$ is  $(c,b\alpha^{2i})$ for some $i$.
When $t$ is even, the center of $E(F_t)$ is $(-c,b\alpha^{2i+1})$ for some $i$.
\end{lemma}
\begin{proof}
Recall that $F_t = A \cap B \cap C \cap H_t$. First, consider odd $t$, where $H_t = R_i$.
We define a map $f\colon (x,y) \to (x, y/\alpha^{2i})$, which rescales the $y$-coordinate. Note that it maps
$F_t$ to $F_1$.
Moreover, $f$ preserves ratios between volumes, and therefore the map of the
smallest ellipsoid containing $F_t$ is the smallest ellipsoid containing
$F_1$. We know that its center is at $(c, b)$, and therefore
the center of $E(F_t)$ is at $(c, b\alpha^{2i})$.

For $t$ even, we have $H_t = L_i$ for some $i$.
We define $g \colon (x,y) \to (-x,y/\alpha^{2i+1})$, which
first mirrors $F_t$ with respect to the $y$ axis, and then rescales the $y$-coordinate, so that $g(F_t) = F_1$.
Clearly, mirroring preserves the volumes, while rescaling preserves
their ratios. Therefore, $f$ maps $E(F_t)$ to $E(F_1)$ whose center is
at $(c, b)$ and the center of $E(F_t)$ is at $(-c, b\alpha^{2i+1})$.
\end{proof}

Let us now estimate the competitive ratio of Algorithm~\ref{alg:ellips}.
At each time step, it incurs cost at least $2c$, since it moves between two
points with $x$-coordinates $c$ and $-c$ respectively.
Therefore, if $N$ is the total number of requests, the total cost incurred by
Algorithm~\ref{alg:ellips} is at least $N\cdot 2c$, which can be arbitrary large.
On the other hand, the point $(0,0)$ is contained in every $F_t$,
since it belongs to $F_1$ and also to every halfspace $R_i$ and $L_i$.
Therefore, the cost of $\OPT$ is at most 1 and the competitive ratio of Algorithm~\ref{alg:ellips}
is unbounded.

\begin{proof}[Proof of Lemma~\ref{lem:lb_ellipsoid}]
Using a computer algebra system, we computed that for $\alpha=1/2$ the center of
$E(F_1)$ is in $(0.24568, 0.40571)$.
This can be calculated, e.g., using the function \texttt{ellipsoidhull} in R, but similar
functions are also available for Matlab. This shows that $\alpha=1/2$ satisfies
the requirements of the lemma.

The manual computation of $E(F_1)$ for $\alpha=1/2$ is laborious,
but we can still give a formal proof of the existence of a suitable $\alpha$. Let $F_1(\alpha)$ denote $F_1$ with parameter $\alpha$.
Observe that $F_1(0)$ is a triangle with vertices $(1,0), (1,1)$,
and $(-1,0)$; and $F_1(1)$ is a square with vertices $(1,0), (1,1)$,
$(-1,0)$ and $(-1,1)$. Since the center of $E(F_1(\alpha))$ evolves continuously with $\alpha$, it suffices to show
that the center of the smallest
ellipsoid containing the triangle $F_1(0)$ lies strictly to the right of the $y$-axis.
By continuity, this implies that there exists $\alpha >0$ such that
the $x$-coordinate of the center of $E(F_1(\alpha))$ is still strictly positive.

We define the affine map
$f\colon (x,y) \to
\bigl(\begin{smallmatrix}1 & -1\\ 0 &\sqrt{3}\end{smallmatrix}\bigr) (x,y)^\top$.
Let $T = F_1(\alpha)$. This transformation makes $T$ equilateral by first
shearing it to the left to be symmetric with respect to the $y$-axis and then
shrinking the $y$-coordinate.
The smallest ellipsoid containing an equilateral triangle is its circumcircle,
whose center lies in the intersection of its altitudes. Since one of the
altitudes lies on the $y$-axis, the $x$-coordinate of the center of $E(f(T))$ 
is 0, and its $y$-coordinate is strictly positive.
Since $f$ preserves ratios between volumes, we have $f(E(T)) = E(f(T))$.
Therefore, applying $f^{-1}$ to the center of $E(f(T))$, we know that the
center of $E(T)$ has both coordinates strictly positive.
\end{proof}

\paragraph{Lower bound for the centroid algorithm.} 
Similar to the 
Ellipsoid-based algorithm, one can propose an algorithm that moves to the
centroid (center of mass) instead:

\begin{algorithm}[H]
Whenever the current position becomes infeasible:

\hspace{3em}
	Move to the centroid of $F_t$.
\caption{A centroid-based algorithm}
\end{algorithm}

The same requests $F_t$ as above also shows that this algorithm is not competitive either.
In fact, the
analysis here is much easier, as we can compute the centroids using
simple geometry (the input convex bodies can be partitioned into a
right triangle and a rectangle, as seen e.g.~in
Figure~\ref{fig:lb_F4}).

A simple calculation shows that for $\alpha = 1/2$, the centroid of
$F_1$ is $(1/9, 7/9)$. For the convex bodies requested later, the
$x$-coordinate of the centroid will oscillate between $-1/9$ and
$1/9$, again showing that the total distance traveled by the
algorithm can be made arbitrarily large.
 
\section*{Acknowledgments}
We would like to thank S\'ebastien Bubeck, Niv Buchbinder, Anupam Gupta, Guru Guruganesh, Crist\'obal Guzm\'an, and Ren\'e Sitters for several interesting discussions. 

\bibliographystyle{alpha}

\begin{thebibliography}{BGK{\etalchar{+}}15}

\bibitem[ABL{\etalchar{+}}13]{ABLLMRW13}
Lachlan L.~H. Andrew, Siddharth Barman, Katrina Ligett, Minghong Lin, Adam
  Meyerson, Alan Roytman, and Adam Wierman.
\newblock A tale of two metrics: Simultaneous bounds on competitiveness and
  regret.
\newblock In {\em {COLT}}, pages 741--763, 2013.

\bibitem[ABL{\etalchar{+}}15]{ABLLMRW15}
Lachlan L.~H. Andrew, Siddharth Barman, Katrina Ligett, Minghong Lin, Adam
  Meyerson, Alan Roytman, and Adam Wierman.
\newblock A tale of two metrics: Simultaneous bounds on competitiveness and
  regret.
\newblock {\em CoRR}, arXiv:1508.03769, 2015.

\bibitem[ABN{\etalchar{+}}16]{Antoniadis16}
Antonios Antoniadis, Neal Barcelo, Michael Nugent, Kirk Pruhs, Kevin Schewior,
  and Michele Scquizzato.
\newblock Chasing convex bodies and functions.
\newblock In {\em {LATIN}}, pages 68--81, 2016.

\bibitem[AS17]{AS17}
Antonios Antoniadis and Kevin Schewior.
\newblock A tight lower bound for online convex optimization with switching
  costs.
\newblock In {\em Submitted, June 2017}, 2017.

\bibitem[BBM01]{BBM01}
Yair Bartal, B{\'{e}}la Bollob{\'{a}}s, and Manor Mendel.
\newblock A {R}amsey-type theorem for metric spaces and its applications for
  metrical task systems and related problems.
\newblock In {\em Foundations of Computer Science, {FOCS}}, pages 396--405,
  2001.

\bibitem[BCN14]{BCN14}
Niv Buchbinder, Shahar Chen, and Joseph Naor.
\newblock Competitive analysis via regularization.
\newblock In {\em SODA}, pages 436--444, 2014.

\bibitem[BGK{\etalchar{+}}15]{BGKPSS15}
Nikhil Bansal, Anupam Gupta, Ravishankar Krishnaswamy, Kirk Pruhs, Kevin
  Schewior, and Clifford Stein.
\newblock A 2-competitive algorithm for online convex optimization with
  switching costs.
\newblock In {\em {APPROX/RANDOM}}, pages 96--109, 2015.

\bibitem[BK04]{BK04}
Yair Bartal and Elias Koutsoupias.
\newblock On the competitive ratio of the work function algorithm for the
  k-server problem.
\newblock {\em Theoretical Computer Science}, 324(2):337 -- 345, 2004.

\bibitem[BLS92]{BLS92}
Allan Borodin, Nathan Linial, and Michael~E. Saks.
\newblock An optimal on-line algorithm for metrical task system.
\newblock {\em J. ACM}, 39(4):745--763, 1992.

\bibitem[BN09]{BN09}
Niv Buchbinder and Joseph Naor.
\newblock Online primal-dual algorithms for covering and packing.
\newblock {\em Math. Oper. Res.}, 34(2):270--286, 2009.

\bibitem[Bur96]{B96}
William~R. Burley.
\newblock Traversing layered graphs using the work function algorithm.
\newblock {\em J. Algorithms}, 20(3):479--511, 1996.

\bibitem[CL96]{CL96}
Marek Chrobak and Lawrence~L. Larmore.
\newblock Metrical task systems, the server problem and the work function
  algorithm.
\newblock In {\em Online Algorithms, The State of the Art}, pages 74--96, 1996.

\bibitem[FL93]{Friedman93}
Joel Friedman and Nathan Linial.
\newblock On convex body chasing.
\newblock {\em Discrete {\&} Computational Geometry}, 9:293--321, 1993.

\bibitem[FM00]{FM00}
Amos Fiat and Manor Mendel.
\newblock Better algorithms for unfair metrical task systems and applications.
\newblock In {\em Symposium on Theory of Computing}, pages 725--734, 2000.

\bibitem[KP95]{KP95}
Elias Koutsoupias and Christos~H. Papadimitriou.
\newblock On the k-server conjecture.
\newblock {\em Journal of the ACM}, 42(5):971--983, 1995.

\bibitem[MG07]{MG07}
Ji\v{r}\'{i} Matou\v{s}ek and Bernd G{\"a}rtner.
\newblock {\em Understanding and using linear programming}.
\newblock Springer-Verlag Berlin Heidelberg, 2007.

\bibitem[Sit14]{Sit14}
Ren{\'{e}} Sitters.
\newblock The generalized work function algorithm is competitive for the
  generalized 2-server problem.
\newblock {\em {SIAM} J. Comput.}, 43(1):96--125, 2014.

\bibitem[SS06]{SS06}
Ren\'{e}~A. Sitters and Leen Stougie.
\newblock The generalized two-server problem.
\newblock {\em Journal of {ACM}}, 53(3):437--458, 2006.

\end{thebibliography}
{\small \newcommand{\etalchar}[1]{$^{#1}$}

 }

\end{document}